\documentclass{article}

\usepackage{PRIMEarxiv}

\usepackage[utf8]{inputenc} 
\usepackage[T1]{fontenc}    
\usepackage{hyperref}       
\usepackage{url}            
\usepackage{booktabs}       
\usepackage{amsfonts}       
\usepackage{nicefrac}       
\usepackage{microtype}      
\usepackage{fancyhdr}       
\usepackage{graphicx}       
\usepackage{soul}
\usepackage{xcolor}
\usepackage{times}
\usepackage{amsmath}
\usepackage{amsthm}
\usepackage{amssymb}
\usepackage{float}
\usepackage{varwidth}
\usepackage{adjustbox}
\usepackage{arydshln}
\usepackage{enumitem}
\usepackage[section]{placeins}

\usepackage[linesnumbered,ruled,vlined]{algorithm2e}
\newcommand{\reform}{REFORM}
\newtheorem{approach}{Approach}
\newtheorem*{challenge}{Challenge}
\newtheorem{proposition}{Proposition}
\newtheorem{lemma}{Lemma}
\newtheorem{corollary}{Corollary}

\newtheorem{definition}{Definition}
\newtheorem{theorem}{Theorem}
\newtheorem*{note}{Note}
\newtheorem*{remark}{Remark}
\SetCommentSty{mycommfont}
\SetKwComment{Comment}{$\triangleright$\ }{}
\SetKwInput{KwInput}{Input}
\SetKwInput{KwRet}{Return}
\graphicspath{{media/}}     
\restylefloat{table}

\title{REFORM: Reputation Based Fair and Temporal Reward Framework for Crowdsourcing}

\author{
  Samhita Kanaparthy  \\
  Machine Learning Lab, IIIT Hyderabad\\
  Hyderabad, India \\
  \texttt{s.v.samhita@research.iiit.ac.in} \\
  \And
  Sankarshan Damle\\
  Machine Learning Lab, IIIT Hyderabad\\
  Hyderabad, India \\
  \texttt{sankarshan.damle@research.iiit.ac.in} \\
  \And
  Sujit Gujar\\
  Machine Learning Lab, IIIT Hyderabad\\
  Hyderabad, India \\
  \texttt{sujit.gujar@iiit.ac.in} \\
}

\begin{document}
\maketitle

\begin{abstract}
Crowdsourcing is an effective method to collect data by employing distributed human population. Researchers introduce appropriate reward mechanisms to incentivize agents to report accurately. In particular, this paper focuses on Peer-Based Mechanisms (PBMs). We observe that with PBMs, crowdsourcing systems may not be fair, i.e., agents may not receive the deserved rewards despite investing efforts and reporting truthfully. Unfair rewards for the agents may discourage participation. This paper aims to build a general framework that assures fairness for PBMs in temporal settings, i.e., settings that prefer early reports. Towards this, we introduce two general notions of fairness for PBMs, namely \emph{$\gamma$-fairness} and \emph{qualitative fairness}. To satisfy these notions, our framework provides trustworthy agents with additional chances of pairing. We introduce Temporal Reputation Model (TERM) to quantify agents' trustworthiness across tasks. With TERM as the key constituent, we present our iterative framework, REFORM, that can adopt the reward scheme of any existing PBM. We demonstrate REFORM's significance by deploying the framework with RPTSC's reward scheme. Specifically, we prove that REFORM with RPTSC considerably improves fairness; while incentivizing truthful and early reports.  We conduct synthetic simulations and show that our framework provides improved fairness over RPTSC.
\end{abstract}


\section{Introduction}
\emph{Crowdsourcing} is a popular method for \emph{requesters} to collect information from many people who input their data through the internet, sensors, and other data streams. Crowdsourcing systems enable the crowd with diverse expertise to contribute to any outsourced \emph{tasks}, including rating products online, urban sensing, collecting real-world data~\cite{review, MCS, Gom1}. In these systems, to maintain accuracy, researchers devote a large body of work in \emph{incentivizing} agents for truthful data elicitation \cite{ppht, pphu, de1, ppts}. Typically, these tasks do not have access to the \emph{ground truth}. Hence, we cannot verify the correctness of the agents' reports. Towards this, 
researchers introduce \emph{Peer Based Mechanisms} (PBMs). PBMs reward an agent based on its consistency with other often random agents referred to as ``peers". For instance, \emph{Robust Peer Truth Serum for Crowdsourcing} (RPTSC)~\cite{PTSC} evaluates an agent's report against a randomly paired peer's report and rewards it if they match. With this, RPTSC incentivizes agents to exert efforts and report truthfully.

In this work, we observe that PBMs are inherently \emph{unfair} as the agent's reward depends on its consistency with randomly selected peers' reports and not primarily on its efforts. While existing PBMs incentivize efforts and truthful reporting, agents still cannot be perfectly reliable as they may have noisy observations or be malicious. In such a case, a reliable or \emph{trustworthy} agent may not get the reward it deserves from unfair pairings. Thus, we believe \emph{fair} rewards are necessary to ensure the participation of trustworthy agents in crowdsourcing.

We consider a crowdsourcing setting in which the task's requester desires real-time data, i.e., it requires the agents to submit their reports at the earliest. We call such a setting \emph{temporal setting}. Such tasks may comprise real-time data collection (e.g., passenger train timetable~\cite{train}, emergent safety incidents information~\cite{campus}, real-time COVID-19 data~\cite{zhang2020crowdsourcing}). In all such tasks, the data reported \emph{early} is \emph{valuable}. It is natural to assume that the reward should decrease with time for any mechanism to incentivize early reporting. However, such \textit{decay} in reward may encourage the agents, who need longer time, to report randomly than exerting efforts, further aggravating the fairness challenges. In summary, we address the following challenge.

\begin{challenge}
To devise a truthful peer-based reward scheme, which ensures fairness while simultaneously incorporating temporal setting.
\end{challenge}

\noindent\textbf{Our Contributions.} To quantify fairness among different PBMs, we introduce two new notions, $\gamma$-fairness (Definition \ref{defn:gamma}) and \emph{Qualitative Fairness} (Definition~\ref{defn:qualitative}). The first notion captures the proximity of the expected reward a PBM guarantees with its optimal reward. The greater the value of $\gamma$, the fairer is the PBM. We believe that this is the first general notion quantifying fairness in PBMs. The second notion ensures that the agent's reward is proportional to its reputation, i.e., trust that the system places on its report.

To achieve these fairness properties in PBMs, we propose the idea of allowing ``trustworthy" agents with additional pairing chance(s). Intuitively, such additional chances will nullify the penalty a trustworthy agent would have incurred from the unfair pairings. To do so, we need quantification of agents' trustworthiness. In general, crowdsourcing systems deploy reputation models for such quantification~\cite{Gom2,MCS}. However, the existing models do not consider temporal setting. Thus, we desire a reputation model that is resistant to manipulation by agents in temporal setting. Towards this, we propose a novel \emph{Temporal Reputation Model} (TERM), which deploys the \emph{Gompertz} function~\cite{Gom2} to output reputation scores. TERM assigns scores to the agents based on their reports and the time taken to submit. We prove that TERM produces high scores for early and truthful reports (Lemma \ref{claim:claim1}).

Having TERM, we design an iterative framework \reform: \emph{REputation based Fair and tempOral Reward fraMework} for crowdsourcing (Framework 1) which takes in the reward scheme of any existing PBM. To exhibit REFORM's effectiveness, we plug in RPTSC's reward scheme and prove that \reform\ with RPTSC reward is strict Nash incentive compatible (Theorem \ref{thm:IC}), i.e., exerting efforts and reporting truthfully and early is a strict Nash equilibrium.

With the definitions of $\gamma$-fairness and qualitative fairness, we show that REFORM framework with RPTSC reward is significantly fairer than RPTSC (Propositions \ref{ptsc-fair}, \ref{reform-fair}) and also qualitatively fair (Theorem~\ref{thm:qf}). We then perform synthetic simulations to validate our results further and show the ameliorated fairness.

\section{Related Work}
Crowdsourcing systems use reward mechanisms to improve the accuracy of the reports~\cite{pphu,de1,ppts}. Towards this, researchers have proposed PBMs which appropriately incentivize agents by evaluating their reports against other agents.

The earliest PBM, \emph{Bayesian Truth Serum} (BTS)~\cite{BTS}, elicits the reports by asking agents to submit their answers and their \emph{beliefs} about other agents' answers. Huang et al.~\cite{PC} propose \emph{Peer consistency} (PC), which rewards the agent based on its consistency with the peer reports. Here, for all the agents, an optimal strategy is to collude and submit the same report. We refer to such a collusion strategy as \emph{single report strategy}. \emph{Peer Truth Serum} (PTS)~\cite{PTS} is a combined version of PC and BTS but requires prior distribution of answers. However, in many settings, prior about the answers is \emph{not} accessible. Randanovic et al.~\cite{PTSC} present an incentive mechanism, \emph{Robust Peer Truth Serum for Crowdsourcing} (RPTSC), which uses the distribution of reported answers from similar tasks as prior. RPTSC can operate with a small number of statistically independent tasks and is resistant to single report strategy. In multi-task settings, \emph{Correlated Agreement} (CA)~\cite{CA} achieve informed truthfulness with an infinite number of tasks, while \emph{Determinant based Mutual Information} (DMI)~\cite{DMI} is dominantly truthful but requires a batch of agents to solve the same set of tasks in a single round. Though we work on Nash Incentive Compatibility, we relax these strong limitations allowing agents to solve any finite number of tasks.

Despite their advantages, we believe that PBMs are inherently unfair, necessitating fair rewards to ensure the participation of trustworthy agents in crowdsourcing. Researchers are actively looking to achieve fairness in crowdsourcing through mechanism design~\cite{fair1,fair3,DBT,Farm}. For instance, Goel et al. \cite{DBT} present Deep Bayesian Trust, which assures fair rewards to the agents by assuming a few gold-standard tasks. Moti et al. \cite{Farm} propose a fair reward mechanism for localized settings. However, we can not apply the proposed notions to address unfairness in PBMs. To the best of our knowledge, we believe that we are the first to address fairness in PBMs.


\section{Preliminaries}
This section describes our crowdsourcing model, relevant game-theoretic definitions and introduces our novel fairness notions. We tabulate notations used in this paper in the supplementary material for reference.

\subsection{Crowdsourcing Model}
The \emph{requester} of the crowdsourcing system publishes a set of $n$ tasks $\mathcal{T} = \{\tau_{1},\tau_{2},\ldots,\tau_{n}\}$ along with their deadlines $\{\delta_{\tau_{1}},\delta_{\tau_{2}},\ldots,\delta_{\tau_{n}}\}$ on the platform every round. These tasks have a discrete and finite answer space $\mathcal{X}$. Tasks in this setting are statistically independent and a-priori similar \cite{DMI}. The requester assigns tasks to available agents $\mathcal{A} = \{a_{1},a_{2},\ldots,a_{m}\}$, and each task in set $\mathcal{T}$ is assigned to at least two agents. We consider that each agent performs a single task in a round for ease of discussion. Whenever an agent is assigned more than one task, we apply our framework to each task solved by the agent separately. Agent $a_i$ is rewarded $R_i(y_i,t_i)$ based on its report $y_i \in \mathcal{X}$ and the time of report, $t_i$.

\subsubsection{Agents ($\mathcal{A}$)}
We assume that the agents in the system are rational and intelligent. In our setting, the agents solve the assigned task either by exerting high ($e_{H}$) or low ($e_{L}$) efforts. Naturally, the cost of putting high efforts is greater, i.e., $c(e_{H}) > c(e_{L})$. We consider any effort which is insufficient to solve the task as low. If agent $a_i$ does its reasonable best to solve the task, i.e., exerts high effort, it obtains an evaluation $x_i$. Based on their efforts and reporting behavior, an agent $a_i$ has a choice between the following two strategies.

\begin{enumerate}[noitemsep,leftmargin=*]
    \item \emph{Trustworthy Strategy.} Exert high effort ($e_{i} = e_{H}$) and report true evaluation ($y_{i} = x_{i}$) at time $t_i^*$, which is the time taken to solve the task. We refer to an agent choosing this strategy as Trustworthy agent (TA).
    \item \emph{Random Strategy.} Exert low effort ($e_{i} = e_{L}$) and report answer randomly according to its prior distribution at any time. We refer to this category of agents as Random agents (RAs).
\end{enumerate}

\subsection{Rewards (R)}
In PBMs, the agents' reports are evaluated against their peers' reports. Depending on its reward scheme, PBMs reward an agent $a_{i}$ having report $y_i$ with reward, denoted by the function $\mathit{peer\mbox{-}fac}(y_i)$. Agent $a_{i}$ is rewarded when its reports match with that of peer $a_p$'s, i.e., $y_{i}=y_{p}$ and is penalized otherwise. For instance, in PTS the reward is $\mathit{peer\mbox{-}fac}(y_i) = \frac{\mathbb{I}_{y_i=y_p}}{f(y_i)}$. Here, $\mathbb{I}_{y_{i}=y_{p}}$ is an indicator variable\footnote{Equal to 1 if $y_i=y_p$ and 0 otherwise}, and $f(y_i)$ is the frequency function which counts the number of occurrences of report $y_i$. As we consider temporal setting, we want the agents to submit their reports at the earliest. One may note that, for any mechanism to incentivize early reporting, the reward should reduce with an increase in time taken for submitting the report. For instance, crowdfunding mechanisms are shown to incentivize early contributions through ``refunds" that decay with time \cite{damle2018designing, damle2019ijcai, chandra2016crowdfunding}. To incorporate this, we employ a decay factor $\beta(t)$ that decays with time in the reward. In our setting, the reward comprises of (i) peer factor $\mathit{peer\mbox{-}fac}(\cdot)$, and (ii) decay factor $\beta(\cdot)$. In summary, the reward agent $a_{i}$ gets on reporting $y_{i}$ for a task $\tau$ after time $t_{i} \leq \delta_{\tau}$\footnote{Reports submitted after the task deadline are not considered.} is,

\begin{equation}\label{eqn:reward}
    R_i(y_{i},t_{i}) = \mathit{peer\mbox{-}fac}(y_{i}) \times \beta(t_{i})     
\end{equation}

Note that, this incentive structure results in a game among the interested agents. To study the game induced, we now define game-theoretic definitions used in the paper. For this, we define the utility of the agent $a_{i}$ as, $u_i = R_i(y_{i},t_{i}) - c(e_{i})$. Further, let $s = (s_{1},s_{2},\ldots,s_{n})$ be the strategy profile such that agent $a_i$'s strategy $s_{i} = (y_{i}, e_{i},t_{i})$ is a tuple consisting of its report, the effort exerted, and the time taken to report. We use subscript $-i$ for the strategy profile without an agent $a_{i}$. Agent $a_{i}$'s utility for solving the task when all the agents play the strategy profile $s = (s_{i},s_{-i})$ is $u_{i}(s_{i},s_{-i})$. Additionally, let $t_i^*$ be the time taken by the agent $a_{i}$ to solve the task. We denote strategy profile $s^{TS} = (s^{TS}_{1},s^{TS}_{2},\ldots,s^{TS}_{n})$ where $s^{TS}_i = (x_i,e_H,t_i^*)\:\forall i$ when all the agents choose trustworthy strategy. And, let $S = \{(y_{i},e_{i},t_{i}) \: | \: \forall y_{i} \in \mathcal{X}, e_{i} \in \{e_{H},e_{L}\}, \forall t_{i} \}$ be the strategy space. With this, we give the following definition.

\begin{definition}[Nash Incentive Compatible (NIC)] A mechanism is said to be NIC if every agent $a_i$ employing trustworthy strategy maximizes its utility given all the other agents choose trustworthy strategy.
\begin{equation}\label{eqn::nic}
    u_{i}(s^{TS}_{i},s^{TS}_{-i}) \geq u_{i}(s_{i}, s^{TS}_{-i}), \forall s_{i} \in S, \forall a_i \in \mathcal{A}     
\end{equation}
\end{definition}

\begin{note}\nonumber
We say that a NIC is \emph{strict} if the inequality (\ref{eqn::nic}) is strict.
\end{note}

\subsection{Quantifying Fairness in PBMs}
As stated, PBMs in the existing literature evaluate agents against a randomly selected peer and reward them based on their reports. One may observe that such evaluations may result in unfair rewards because of a trustworthy agent getting paired with a randomly selected agent. Even when a PBM incentivizes agents to exert efforts and report truthfully, if still, the agent's peer receives a noisy observation or is malicious, it can lead to unfair rewards for the agent who observed correct signals. As a result, a trustworthy agent may incur an unfair penalty. The lack of fairness in PBMs, i.e., trustworthy agents incurring penalties, was first observed in \cite{peer-unfair}. To quantify fairness, we now present two novel notions of fairness relevant to PBMs.

\subsubsection{$\gamma$-Fairness}
This notion of fairness depends on the difference in optimal and expected rewards of trustworthy agents in a PBM. We believe that this is the first general notion of quantifying fairness in PBMs.
    
Let for any PBM, $M^{*}$ be the optimal reward a trustworthy agent gets when its report matches with a peer's report, and $E^{*}$ be the expected reward. With this, we define $\gamma$-fairness as follows,
    \begin{definition}[$\gamma$-Fairness]\label{defn:gamma}
        For a trustworthy agent, we say a PBM is $\gamma$-Fair if the expected difference in its optimal and the expected rewards is equal to $\frac{1}{\gamma}$, that is, 
        $$\mathbb{E}_{x \in \mathcal{X}}\left[M^{*}-E^{*}\right] = \frac{1}{\gamma}$$
    \end{definition}
    
\subsubsection{Qualitative Fairness}
In mechanisms that deploy reputation scores, it is desirable to prioritize the agents with better reputation over an agent with a lesser reputation. A report from the agent who promptly submits the truth is always valuable compared to an agent's report with an arbitrary history of reporting.
    
Similarly, in PBMs with reputation scores, an agent with a higher reputation should have a higher expected reward than agents with the same report but a lesser reputation. We capture this desired property with a new notion of fairness, namely \emph{Qualitative Fairness}.
    
    \begin{definition}[Qualitative Fairness]\label{defn:qualitative}
        Let agents $a_{i}$, $a_{j}\in \mathcal{A}$ submit their reports $y_{i}$, $y_{j}$ at the same time $t$ such that $y_{i}=y_{j}$. We say a PBM guarantees qualitative fairness if its rewards satisfy,
        \begin{equation}\nonumber
            \mathbb{E}[R_{i}(y_{i} = y,t)|\Omega_{i}] \geq \mathbb{E}[R_{j}(y_{j} = y,t)|\Omega_{j}] \\
                \quad \forall\Omega_{i} \geq \Omega_{j}, \forall y \in \mathcal{X}, \forall i,j.        
        \end{equation}
    \end{definition}
Here, $\mathbb{E}[R_{i}(y_{i} = y,t)|\Omega_{i}]$ is expected reward of agent $a_{i}$ with reputation score $\Omega_i$ for reporting $y_i$ at time $t$.
    
\subsection{Our Approach: Ensuring Fairness in PBMs}\label{ssec:approach}
To improve fairness in PBMs and satisfy the given notions, we look at two different approaches and analyze them.

    \begin{approach}\label{app:1}
        A straightforward way to overcome unfairness in PBMs is by using average (or weighted average) over multiple reports, say $w$ reports, to reward a particular report. 
    \end{approach}
    
This reduces the penalty obtained from unfair pairings for a trustworthy agent to some extent. However, it also assures higher expected rewards for random agents, increasing the overall budget, which is not desirable. Thus, we aim for a reward scheme that guarantees better fairness for trustworthy agents while discouraging random reporting. Towards this, we present our approach.
        
    \begin{approach}\label{app:2}
        We provide agents with additional chances of pairing, say $k$, to evaluate their reports when paired with a less reputed agent.
    \end{approach} 
    
Consider a trustworthy agent who reported $y$ for some task. Let the optimal reward it obtains when its report matches its peer's report (ignoring the decay factor $\beta(\cdot)$) be $g = \mathit{peer\mbox{-}fac}(y|y'=y)$ and penalty obtained when reports do not match be $l = \mathit{peer\mbox{-}fac}(y | y'\neq y)$.

\noindent\textbf{Expected reward with Approach \ref{app:1}}: Let us consider that the probability of an agent's report matching with its random peer's report be $\frac{1}{2}$. Now, averaging reward across $w$ such reports, the expected reward of an agent is
    \begin{equation}\nonumber
        R\mbox{-}1  = \frac{1}{w}\left(\frac{1}{2}\left(w\cdot g\right) + \frac{1}{2}\left(w \cdot l\right)\right) =  \frac{1}{2} g + \frac{1}{2} l
    \end{equation}

\noindent\textbf{Expected reward with Approach \ref{app:2}}: For our approach, consider the least number of additional chances of pairing, i.e., $k=2$. The agent is given another chance when its report does not match, and its reputation score is higher than its peer's in the first matching. Here, we assume that an agent's reputation is more than its peer with probability $0\leq r \leq 1$. Hence, the expected reward is
    \begin{equation}\nonumber
        R\mbox{-}2  = \frac{1}{2}g + \frac{1}{2}\left((1-r)l + t \left(\frac{1}{2}g + \frac{1}{2}l\right)\right) =  \left(\frac{1}{2} + \frac{r}{4}\right) g + \left(\frac{1}{2} - \frac{r}{4}\right) l
    \end{equation}

As we see, the expected reward $R\mbox{-}2$ is greater and is closer to the optimal reward $g$ compared to $R\mbox{-}1$. Hence, we adopt Approach \ref{app:2}, which guarantees better fairness compared to Approach \ref{app:1}.

\begin{remark}
    The ingenuity of our approach is to give trustworthy agents additional chances of pairing to evaluate their reports, which reduces the possibility of agents getting penalized for unfair pairings. This decrease in unfair penalty leads to higher expected rewards, further improving fairness. However, we must ensure that this process does not increase the expected rewards for agents with random strategies due to additional matching. To decide which agent will receive additional chances to pair, we use \emph{reputation scores} as a metric.
\end{remark}

Research has shown that reputation scores successfully quantify the trust a crowdsourcing system must place on an individual agent based on its history. However, no reputation model exists in the literature, which factors in the time taken to submit the report. Consequently, we introduce our reputation model, \emph{TERM}, to quantify trustworthiness in temporal setting in the next section.

\section{Temporal Reputation Model (TERM)}\label{ssec:term}
In crowdsourcing systems, mechanisms introduce \emph{reputation score} of an agent as a parameter of trust the system places in its submitted report. The system builds up this \emph{trust} in the agent, gradually after several instances of trustworthy behavior, and diminishes relatively quickly if the system observes adversarial behavior. This logic applies to any trust-based system such as Amazon Mechanical Turk~\cite{AMT}, Crowdflower~\cite{CrowdFlower}, etc.

Typically, reputation scores require to satisfy the following:
\begin{enumerate}[leftmargin=*]
    \item Builds trust in the agent gradually with honest behavior.
    \item To incorporate temporal setting, the increase in scores should be inversely proportional to the time taken to report.
    \item The score growth should decrease as it reaches the extreme and should not cross the maximum score allowed.
\end{enumerate}

In general, the existing literature for a reputation model in crowdsourcing only factors the report submitted by the agents~\cite{Gom2, softpenality, MCS}. Towards this, we propose \emph{Temporal Reputation Model} (TERM), which assigns TERM scores to agents considering both the accuracy of the report and the time taken to submit.

\subsection{Computation of TERM scores}
For TERM to satisfy the properties mentioned above, we use \emph{Gompertz function}~\cite{Gom2} whose variation is gradual, smooth, and is well suited for the model. Gompertz function is a particular case of sigmoid function, in which the growth at the start and end is slow. Several crowdsourcing mechanisms deploy this function to measure trust~\cite{Gom2, Gom1, Gom3}.
In TERM, we maintain the agents' history. 

\smallskip
\noindent\textbf{History ($\mathcal{H}$).} We maintain a history $\mathcal{H}$ of all the scores for every agent in each round. Let $\mathcal{H}_{i,j} = (\Omega_{i,j}, |\phi|_{i,j}, |\phi|_{i,j-1}, \ldots, |\phi|_{i,1})$ denote the history of agent $a_{i}$ till the round $r_{j}$, where $\Omega_{i,j}$ is the TERM score, $|\phi|_{i,j}$ is the normalized round-score obtained for the report submitted in the round $r_{j}$.

Algorithm~\ref{algo::Term} presents TERM, here, the  requester  maintains  frequency $f(y_i)$ of the report $y_i$. TERM calculates normalized round-scores of agents from the reports submitted and time taken for reporting (Lines 5-6). The cumulative-score calculation uses all the obtained normalized round scores until the latest round (Line 7). We take cumulative-score as input to the Gompertz function, whose output is the TERM score (Line 8). We now formally define TERM and give its properties.

Let $\Omega_{i,j}$ be the TERM score an agent $a_{i}$ obtains after round $r_{j}$. We define TERM score as,
\begin{equation}\label{eqn:func}\tag{TERM}
    \Omega_{i,j} = G({\psi}_{i,j}) = a\times \exp(b \times \exp(c \times{\psi}_{i,j}))
\end{equation}

where, $G(\cdot)$ is the Gompertz function with parameters $a \in \mathbb{R}$ controls the asymptote, $b \in \mathbb{R^{-}}$ sets the displacement, and $c \in \mathbb{R^{-}}$ controls the growth rate of the curve. We set $a=1$, $b=-1$, $c=-1/2$ for a smooth growth of TERM scores. Here, the input to Gompertz function is cummulative score ${\psi}_{i,j}$ (defined in Eq. \ref{eqn:cumscore})

Further, ${\phi}_{i,j}$ is \emph{round-score} obtained by agent $a_{i}$ for submitting $y_{i}$ in the round $r_{i}$ after time $t_{i}$, defined as,
\begin{equation}\label{eqn:r-s}
    {\phi}_{i,j} = \frac{\mathbb{I}_{y_{i} = y_{p}}}{f(y_{i})t_{i}}
\end{equation}
Where $y_{p}$ is the random peer's report chosen from the same task. 

We \emph{map} all the round-scores to $[-1,1]$ as follows, 

\begin{equation} \nonumber
    |\phi|_{i,j} =
    \begin{cases}
        \frac{{\phi}_{i,j} - min({\phi}^{j})}{max({\phi}^{j}) - min({\phi}^{j})}, & \text{if $max({\phi}^{j}) \neq min({\phi}^{j})$} \\
        0, & \text{otherwise}
    \end{cases}
\end{equation}

Here, $min({\phi}^{j})$ and $max({\phi}^{j})$ denote that minimum and maximum round-scores in the round $r_{j}$, respectively.

We then calculate \emph{cumulative-score} ${\psi}_{i,j}$ of each agent $a_i$ by taking all the normalized round-scores ($|\phi|_{i,j}$) it obtained till the latest round, as follows,

\begin{equation}\label{eqn:cumscore}
    {\psi}_{i,j} = \sum_{k=1}^{j} \lambda^{(j-k)} {|\phi|}_{i,k} \quad(0<\lambda<1)
\end{equation}
Trivially, $\lambda^{(j-k)}$ gradually reduces the impact of previous round-scores. As desired, TERM score is an aggregate of \emph{all} the submissions made by an agent and accounts for the fact that recent submissions are more relevant.

\begin{small}
\begin{algorithm}[!t]
Agent $a_{i}$ submits report $y_{i}$ for a task $\tau$ in round $r_{j}$ at time $t_{i}$.\\
\textbf{Input}: Report $y_{i}$, Time taken $t_{i}$, History $\mathcal{H}_{i,j-1}$\\
\textbf{Output}: Updated TERM score $\Omega_{i,j}$\\
Randomly choose a report $y_{p}$ of agent $a_p$ from the same task $\tau$.\\
${\phi}_{i,j} = \frac{\mathbb{I}_{y_{i} = y_{p}}}{f(y_{i})t_{i}}$ \Comment*[r]{round-scores calculation}
$|\phi|_{i,j} \leftarrow $ normalised ${\phi}_{i,j}$;\\
${\psi}_{i,j} = \sum_{k=1}^{j} \lambda^{(j-k)} {|\phi|}_{i,k}$ \Comment*[r]{cumulative-scores calculation}
\textbf{Return}: $\Omega_{i,j} = \exp(-\exp(\frac{-{\psi}_{i,j}}{2}))$ \Comment*[r]{TERM score calculation}
\caption{\label{algo::Term} $\text{TERM}(y_{i},t_{i},\mathcal{H}_{i,{j-1}})$}
\end{algorithm}
\end{small}

\subsubsection{TERM Score properties}  
Notice that Eq. \ref{eqn:func} used in the reputation model gradually increases with early reporting but reduces relatively fast with random reporting when the reports do not match. With this, one can observe that trustworthy reporting benefits the agents over random reporting. Hence, agents cannot manipulate their TERM score. 

We now consider a \emph{collusive} strategy wherein all agents collude to submit the same report, i.e., single report strategy. We prove that TERM score is resistant to such a strategy.
\begin{lemma}\label{claim:claim3}
    \emph{TERM} is resistant to single report strategy. 
\end{lemma}


In the next section, we introduce our novel iterative framework \reform\ for crowdsourcing which uses TERM as a key component.

\section{REFORM: Framework}

We now present \reform\, a novel iterative framework for crowdsourcing, based on Approach \ref{app:2} (Section~\ref{ssec:approach}). 
Intuitively, \reform\ incentivizes an agent to report truthfully by improving the expected reward of trustworthy agents. We achieve this increase in the expected reward by ingeniously offering reputed agents additional chance(s) of pairing. 
Framework 1 formally presents REFORM. Observe that in Framework 1, $\mathit{peer\mbox{-}fac}(\cdot)$ may be the reward scheme of any existing PBM. 
In \reform, based on the selection of the reward scheme, we evaluate an agent's report against the randomly chosen peer's report from the same task and reward the agent if the reports match (Lines 7-8). 
For temporal setting, we use TERM scores to decide whether to offer additional chance(s) of pairing to an agent. More concretely, if an agent's submitted report does not match with its peer's report, and if the agent has a TERM score \emph{higher} than that of its peer without having reached the maximum number of chances $k$, we give it another chance to pair. Otherwise, we penalize according to the reward scheme adopted (Lines 10-11).

\smallskip
\noindent\textit{Note.} We can plug any relevant reputation model instead of TERM in REFORM, and it still helps to improve the fairness of a PBM. As we focus on crowdsourcing in temporal setting, we employ our novel reputation model TERM for the same.

\begin{algorithm}[!t]
\renewcommand{\algorithmcfname}{Framework}
Agent $a_{i}$ submits a report $y_i$ for an assigned task $\tau$ at time $t_i \leq \delta_{\tau}$ in round $r_j$.\\
\KwInput {$\mathit{peer\mbox{-}fac}(\cdot)$, $k > 1$, $y_i$, $t_i$, History $\mathcal{H}_{i,j-1}$} 
\textbf{initialization}: $l=0$\\
 \While{$l<k$}{
   Randomly choose peer report $y_{p}$ from the same task $\tau$.
 \If{$l = 1$}{$\Omega_{i,j} = \text{TERM}(y_{i},t_{i},\mathcal{H}_{i,j-1})$ \Comment*[r]{update TERM score}}
  $l = l + 1$\\
  \If{$y_{i} = y_{p}$}
  {\tcc{reports match, agent gets optimal reward}
  {\textbf{Return: } $R_i(y_{i},t_{i}) = \mathit{peer\mbox{-}fac}(y_i|y_i = y_p) \times \beta(t_i)$} \\
  }
  \Else
  {
   \If {$\Omega_{i,j} \leq \Omega_{p,j} \lor l=k$}
   { \tcc{reputation score is less or maximum chances reached, no more pairing}
  \textbf{Return: } $R_i(y_{i},t_{i}) = \mathit{peer\mbox{-}fac}(y_i|y_i \neq y_p) \times \beta(t_i)$
  }
  }
 }
 \caption{\label{algo::reform}\reform}
\end{algorithm}

We next demonstrate REFORM’s significance by deploying the framework over RPTSC’s reward scheme. We first present the agent beliefs required in REFORM with RPTSC reward and its properties relevant to our work. Later we demonstrate its game-theoretic and fairness properties.

\subsection{REFORM with RPTSC reward scheme}

We focus on RPTSC~\cite{PTSC} than other PBMs because of its practicality. 
Briefly, RPTSC is desirable over other PBMs as it (i) does not assume any prior, (ii) incentivizes efforts and truthful reporting, (iii) uses a \emph{surprisingly common rule} to reward, and (iv) is resistant to single report strategy. The mechanism evaluates agents’ reports against a randomly chosen peer’s report, and each report $y_{i}$ is rewarded as follows:
    \begin{equation}\label{rptsc-reward}
        {\mathit{peer\mbox{-}fac}(y_i)}_{RPTSC} = 
        \begin{cases}
            \alpha\left(\frac{\mathbb{I}_{y_{i}=y_{p}}}{f(y_i)}-1\right), & \text{if }f(y_i) \neq 0 \\
            0, & \text{otherwise}
        \end{cases}
    \end{equation}
where, $\alpha > 0$ is a scalar constant used to tweak the reward budget. Further, $\mathbb{I}_{y_{i}=y_{p}}$ is an indicator variable and $f(y_i)$ is a frequency function, as defined previously.

\subsubsection{Agents Beliefs}
In our setting, every agent has their own private belief about other agents' reports and reputations. Thus, the agent's estimate of its expected reward depends on its belief. We identify the distributions, $P(\cdot)$, $Q(\cdot)$, and $T(\cdot)$, as agents' beliefs about evaluations, reports, and reputation scores, respectively; as defined next.

\subsubsection*{Agent Beliefs about Evaluations ($P$)}
The prior belief $P_{p}(x_p)$ denotes the probability with which agent $a_{p}$'s evaluation is $x_p$.
Consider $a_p$ as another agent who solves the same task.
Then, $P_{p|i}(x_p|x_i)$ is $a_i$'s posterior belief about  agent $a_{p}$'s evaluation being $x_p$ when its evaluation is $x_i$.
We assume that all the agents' beliefs are \emph{fully mixed}, i.e., they believe that every possible event has a minimum probability (strictly $> 0$) of occurrence. Mathematically, for all the agents,
$$
\forall x_i,x_p \in \mathcal{X}: 0 < P_{p}(x_p), P_{p|i}(x_p|x_i) <1.
$$
    
Since the tasks are statistically independent, if agent $a_{i}$ has not solved task $\tau_k$, it has no evaluation for that task, i.e., $x_{i}=\varnothing$. Therefore, agent $a_i$'s posterior belief about the evaluation of agent $a_q$ for task $\tau_k$ is $P_{q|i}(x_q|x_{i}) = P_{q|i}(x_q|\varnothing) = P_{q}(x_q) \forall x \in \mathcal{X}$. That is, the same as its prior belief.

\subsubsection*{Agent Beliefs about Reports ($Q$)}
To decide its best strategy, agent $a_i$ estimates its expected reward for reporting $y_{i}$ based on its beliefs about the other agents' reports. For this, it transforms its beliefs about others' evaluations into beliefs about reports; $(P_{p}, P_{p|i}) \to (Q_{p}, Q_{p|i})$. The posterior belief $Q_{p|i}(y_p|x_i)$ is the probability that agent $a_{p}$ reports $y_p$ when agent $a_i$'s evaluation is $x_i$. As random agents do not evaluate the task; they do not have posterior beliefs. They report according to their prior distribution. Formally, if all the agents choose random strategy, i.e., their evaluations are $\varnothing$, we have $Q_{p}(y) = Q_{q}(y) = P_{p}(y)$ and $Q_{p|i}(y|\varnothing) = P_{p|i}(y|\varnothing) = P_{p}(y) \forall y \in \mathcal{X}$. Likewise, when all the agents choose trustworthy strategy, i.e., all of them report their true evaluations, their beliefs about evaluations are the same as that of reports, $Q_{p}(y) = Q_{q}(y)=P_{p}(y)$; $Q_{p|i}(y)=P_{p|i}(y) \forall y \in \mathcal{X}$.


\noindent\textbf{Self-predicting Condition.}
If agent $a_{i}$ exerts high effort ($e_H$) for the task, to acquire an evaluation $x_i$, it develops a posterior belief $P_{p|i}(y_i|x_i)$ regarding the evaluations of peers. We assume that the posterior belief has a positive correlation with agent $a_i$'s evaluation $x_i$. This assumption is the \emph{self-predicting} condition defined as~\cite{PTSC}:
    $$
    \frac{P_{p|i}(x_i|x_i)}{P_{p}(x_i)} > \frac{P_{p|i}(y_i|x_i)}{P_{p}(y_i)} 
    $$

If an agent $a_{i}$'s belief satisfies the self-predicting condition, RPTSC~\cite{PTSC} introduces a \emph{self-predictor} $\Delta_{i}$ as a minimum value in  $[0,1]$ such that the following holds.
    \begin{equation}\label{eqn:cond1}
    \left(\frac{P_{p|i}(x_i|x_i)}{P_{p}(x_i)} - 1\right)\Delta_{i} > \frac{P_{p|i}(y_i|x_i)}{P_{p}(y_i)} - 1     
    \end{equation}

We redefine $\Delta_{i}$ as the infimum over all possible $\Delta_i \in [0,1]$ satisfying:
    \begin{equation}\label{eqn:A1}
        \left(\frac{P_{p|i}(x_i|x_i)}{P_{p}(x_i)}\right)\Delta_i > \frac{P_{p|i}(y_i|x_i)}{P_{p}(y_i)}
    \end{equation}

Self-predictor $\Delta_i$ characterizes an agent's belief regarding the degree of correlation of its evaluation with others' reports. The smaller the value of $\Delta_{i}$, greater the belief agent $a_i$ has on its evaluation. If $\Delta_{i}\approx 1$, the agent $a_i$ is more likely to confuse between different answers. In contrast, if $\Delta_{i}\approx 0$, the answers do not correlate. 

\subsubsection*{Agents Beliefs about Reputation Scores ($T$)}
We recognize an agent $a_i$'s beliefs regarding the reputation score of agent $a_{p}$ with a distribution ${T_{p}}$. $T(\cdot)$ is a distribution that is \textit{symmetric} and \textit{identical}. The distribution does not change with a particular peer, and all the agents share the same distribution, that is, 
        $$
        T_p=T_{q}; ~\forall a_p,a_q \in \mathcal{A}.
        $$ 
        
Particularly, when the underlying reputation score is TERM (Eq.~\ref{eqn:func}), we denote $T_{p}(\Omega_i,\Omega_p)=Pr(\Omega_i \geq \Omega_p)$ as the probability that an agent $a_p$ has its TERM score less than $\Omega_i$.

\subsubsection{RPTSC Properties} We now summarize the relevant properties of RPTSC required for the analysis of our results. 
To ease our notations, we use the following set of notations in the remainder of the paper.

\begin{itemize}[noitemsep,leftmargin=*]
    \item $p_{y_{i}} := P_{p}(y_{i})$;   $p'_{y_{i}} := P_{p|i}(y_{i}|x_{i})$; $p_{x_{i}} := P_{p}(x_{i})$
    \item $p'_{x_{i}} := P_{p|i}(x_{i}|x_{i})$; $q_{y_{i}} := Q_{p}(y_{i})$; $q'_{y_{i}} := Q_{p|i}(y_{i}|x_{i})$
    \item $q_{x_{i}} := Q_{p}(x_{i})$; $q'_{x_{i}} := Q_{p|i}(x_{i}|x_{i})$; $r = T_{p}(\Omega_i,\Omega_p)$
\end{itemize}


\begin{proposition}\label{prop:p1} \textup{\cite[Lemma 4.1]{PTSC}}
In RPTSC, the expected reward of agent $a_{i}$ with evaluation $x_{i}$ and report $y_{i}$ is,
\begin{equation}
    \nonumber
    E'=
    \begin{cases}
        \alpha\left(\frac{q'_{y_{i}}}{q_{y_{i}}} - 1\right)\left(1 - \left(1 - q_{y_{i}}\right)^{n-1}\right), & \text{if }q_{y_{i}}>0\\
        0, & otherwise
    \end{cases}
\end{equation}  
\end{proposition}

Let $M'$ be the optimal reward of a trustworthy agent $a_{i}$ in RPTSC (i.e., if its report matches its peer's report, $q'_{y_{i}} = 1$ in Prop. \ref{prop:p1}). This implies that,
    \begin{equation}\label{eqn:optreward}
        \begin{split}
            M' &= \alpha \left(\frac{1}{q_{y_{i}}} - 1\right)\left(1 - \left(1 - q_{y_{i}}\right)^{n-1}\right).
        \end{split}
    \end{equation}

In RPTSC, the expected reward of an agent $a_{i}$ before evaluation when all the other agents choose trustworthy strategy is, from~\cite[Eq. 8]{PTSC}, 
\begin{equation}\label{eqn:R'} 
  \overline{R}_{i}(\alpha) = \mathbb{E}_{x_i \in \mathcal{X}}\left[\alpha\left(\frac{p'_{x_{i}}}{p_{x_{i}}} - 1\right)\left(1 - \left(1 - p_{x_{i}}\right)^{n-1}\right)\right],
\end{equation}

\begin{proposition}\label{prop:p2} \textup{\cite[Theorem 4.3]{PTSC}}
    Trustworthy strategy is strict Nash Equilibrium in RPTSC, if all the agents $a_{i}$, answers space $\mathcal{X}$ , scaling factor $\alpha$ and number of tasks $n$ satisfy:
    \begin{equation}
        \begin{split}
            &\mathsf{A} : \overline{R}_{i}(\alpha) > c(e_{H}) - c(e_{L}) \\
            &\mathsf{B} : \frac{1-(1-p_{x_{i}})^{n-1}}{1-{p_{x_{i}}}^{n-1}} \geq \Delta_{i} \nonumber
        \end{split}
    \end{equation} 
\end{proposition}

To simplify the proofs, we use Assumption $\mathsf{B}_1$, which is obtained from Assumption $\mathsf{B}$ by putting $n:=n+1$.
\begin{equation}\label{eqn:A2}
    \mathsf{B}_1 : \frac{1-(1-p_{x_{i}})^{n}}{1-{p_{x_{i}}}^{n}} \geq \Delta_{i}    
\end{equation}


\begin{algorithm}[!t]
\SetAlgoLined
 agent $a_{i}$ submits a report for an assigned task $\tau$ in round $r_{j}$.\\
\KwInput {$k > 1$, report $y_{i}$, time taken $t_{i} \leq \delta_{\tau}$, History $\mathcal{H}_{i,j}$}
\KwResult{Reward, $R_i(y_{i},t_{i})$}
 \textbf{initialization}: $l=1$. Randomly sample $n-1$ reports, one from each task other than $\tau$\\
 \While{$l<k$}{
   Randomly choose another report $y_{p}$ from the same task $\tau$. Calculate the frequency of $y_{i}$ from the sampled reports, as $f(y_{i}) = \frac{num(y_{i})}{\sum_{y\in \mathcal{X}}num(y)}$\\
 \If{$l = 1$}{$\Omega_{i,j} = \text{TERM}(y_{i},t_{i},\mathcal{H}_{i,j})$ \Comment*[r]{update TERM score}}
  $l = l + 1$\\
  \If{$y_{i} = y_{p}$}
  {
  {\textbf{Return: }$R_i(y_{i},t_{i}) = \alpha\beta(t_{i})\left(\frac{1}{f(y_{i})} - 1\right)$} \Comment*[r]{reports match, algorithm ends}
  }
  \Else
  {
   \If {$\Omega_{i,j} \leq \Omega_{p,j} \lor l=k$}
   { \tcc{TERM score is less or maximum chances reached, no more pairing}
   \If{$f(y_{i}) \neq 0$}
   {
   \textbf{Return: }$R_i(y_{i},t_{i}) = -\alpha\beta(t_{i})$ \Comment*[r]{obtains penalty, algorithm ends}
   }
   \Else{
   \textbf{Return: }$R_i(y_{i},t_{i}) = 0$ \Comment*[r]{frequency of report is 0, algorithm ends}
   }
  }
  }
 }
 \caption{\label{algo::reform-rptsc}\reform\ with RPTSC}
\end{algorithm}

\subsection{REFORM with RPTSC: Game-Theoretic Analysis}

We now present the essential properties and results for REFORM with RPTSC reward scheme. Crucially, our results \textit{do not} require any further assumptions, over the standard RPTSC assumptions, i.e., $\mathsf{A}$ and $\mathsf{B}_1$. Due to space constraints, we omit the formal proofs of the results. The corresponding proofs are available in the supplementary material.

Within RPTSC assumptions, TERM scores are high for truthful and early reporting. Observe that TERM scores increase with cumulative-score, which aggregates all the normalized round scores. Hence, an increase in round-score increases the TERM score. Having the self-predicting condition Eq.~\ref{eqn:A1}, round-scores are always higher for truthful reports. With this, we give the following lemma.

\begin{lemma}\label{claim:claim1}
        Given all the other agents choose trustworthy strategy, \emph{TERM} incentivizes an agent to report truthfully and early under Assumption $\mathsf{B}_1$ $(\emph{Eq.}~\ref{eqn:A2})$.
\end{lemma}

Having TERM's incentive properties, we now provide the expected reward of RPTSC in \reform\ framework.

\begin{lemma}\label{lemma:expreward}
In \reform\ with RPTSC, the expected reward of agent $a_{i}$ with evaluation $x_{i}$ and report $y_{i}$ at time $t_{i}$ is, $\mathbb{E}[R_i(y_{i}|x_{i});k=2] =$
    \begin{equation}
    \begin{cases}
        \alpha\beta(t_{i})\left[\frac{q'_{y_{i}}}{q_{y_{i}}}-1 + r(1-q'_{y_{i}})\frac{q'_{y_{i}}}{q_{y_{i}}}\right]\left[1-(1-q_{y_{i}})^{n-1}\right], & \text{if }q_{y_{i}} > 0\nonumber\\
        0, & \text{otherwise}
    \end{cases}
\end{equation}
Where $k=2$ and $n$ is the number of tasks.
\end{lemma}
\begin{proof} (SKETCH).
    Trivially, if $q_{y_{i}} = 0$ the expected reward is $0$. For $q_{y_{i}} \neq 0$, we calculate the expected reward as follows: Observe that when $k=1$, \reform\ with RPTSC has similar reward as RPTSC, and hence, the expected reward for $k=1$ is the decay factor times the expected reward in RPTSC, i.e., $\beta(t_{i})\times E'$ (refer Prop.~\ref{prop:p1}).

    When $k=2$, if the agent's TERM score is less than that of its peers', there is no additional chance of pairing; therefore, the expected reward is the same as $\beta(t_{i})\times E'$. 
    However, if the agent's TERM score is high, the agent gets an additional chance. In this case, 
    (i) If reports match in the first pairing, the agent gets a reward which is the decay factor times the optimal reward $M'$ (Eq.~\ref{eqn:optreward}), i.e., $\beta(t_{i})\times M'$.
    (ii) Otherwise, the agent can get an additional chance, and the reward is $\beta(t_{i})\times E'$. All the cases mentioned above, calculated with their probabilities of occurrence, give the expected reward an agent obtains in \reform\ framework with RPTSC reward.
\end{proof}

Note that the expected reward in REFORM with RPTSC is greater than RPTSC's expected reward (Prop. \ref{prop:p1}). This shows the desirability of \reform\, i.e., the agents with higher TERM scores (i.e., TAs) are benefited significantly more in our framework. Naturally, the increase in the expected reward will increase as the number of chances for pairing. We formalize this intuitive result in Corollary~\ref{cor::k}.
    
\begin{corollary}\label{cor::k}
In \reform, the expected reward increases with an increase in additional chances, $k$.
\end{corollary}

\subsubsection*{REFORM with RPTSC: Nash Incentive Compatibility.}
  We now prove that in \reform\ with RPTSC's reward, it is strict Nash equilibrium for every agent to choose trustworthy strategy.
  To report truthfully, each agent must first exert high efforts. To incentivize agents to exert efforts, their expected utility must be strictly greater than randomly reporting (by investing low efforts). Considering all the agents are trustworthy (i.e., $q'_{y_{i}} = p'_{x_{i}}, q_{y_{i}} = p_{x_{i}}$), agent $a_{i}$'s expected reward for reporting at time $t_{i}$, before its evaluation is,
\begin{equation}\label{eqn:R}
    \begin{split}
        \overline{Ref}_{i}(\alpha) = &\alpha \beta(t_{i}) \mathbb{E}_{x_i\in \mathcal{X}} \bigg[\left(\frac{p'_{x_{i}}}{p_{x_{i}}} - 1 + rp'_{x_{i}}\frac{(1-p'_{x_{i}})}{p_{x_{i}}}\right)\\
        & \quad \left(1 - \left(1 - p_{x_{i}}\right)^{n-1}\right)\bigg].
    \end{split}
\end{equation}

We provide a random agent's expected reward, $E_{ra}$, in Lemma~\ref{lemma:random}. We then prove that \reform\ with RPTSC incentivizes agents to exert efforts by showing that the expected utility of an agent for exerting high efforts is strictly greater than random reporting (Lemma~\ref{lemma:IC1}).

\begin{lemma}\label{lemma:random}
    In \reform\ with RPTSC, the expected reward of a random agent $a_{i}$ with report $y_{i}$ at time $t_{i}$ when all the other agents choose trustworthy strategy is, \emph{$E_{ra}$}= $\alpha r \beta(t_{i})(1-p_{y_{i}}) \left(1 - \left(1 - p_{y_{i}}\right)^{n-1}\right)$. 
\end{lemma}

To ensure that the agents who exert efforts are accordingly compensated, we choose the scalar constant $\alpha$ in RPTSC's reward such that \textit{it absorbs the delay caused due to efforts}. For agent $a_{i}$ who reported at time $t_{i}$, $\alpha$ should satisfy:\footnote{Assumption $\mathsf{A}_1$ is a modification of Assumption $\mathsf{A}$ to make it fit for temporal setting.}
\begin{equation}\label{eqn:R1}
    \mathsf{A}_1: \overline{Ref}_{i}(\alpha) - \alpha \beta(t_{i}) \geq c(e_{H}) - c(e_{L}).
\end{equation}

\begin{lemma}\label{lemma:IC1}
    In \reform\ with RPTSC, an agent is incentivized to exert high efforts given all the other agents choose trustworthy strategy under Assumption $\mathsf{A}_1$ \emph{(\text{Eq. }\ref{eqn:R1})}.
\end{lemma}
\begin{proof} (SKETCH) We know that the expected utility of a trustworthy agent $a_{i}$ for exerting efforts, before evaluation of tasks is $\overline{R}_{i}(\alpha) - c(e_{H})$ (from Eq.~\ref{eqn:R}). Moreover, the expected utility of a random agent is $E_{ra} - c(e_{L})$ (from Lemma~\ref{lemma:random}). Now from Assumption $\mathsf{A}_1$ (Eq.~\ref{eqn:R1}), we have $\overline{Ref}_{i}(\alpha) \geq c(e_{H})-c(e_{L}) + \alpha \beta(t_{i})$. The proof then follows by showing that $E_{ra} < \alpha \beta(t_{i})$. This implies $\overline{R}_{i}(\alpha) - c(e_{H}) \geq E_{ra} - c(e_{L})$, that is, the expected utility of an agent for exerting efforts is greater than reporting randomly.
\end{proof}

From Lemma~\ref{lemma:IC1}, one can note that random agents are at a disadvantage. What is more, TERM score for agents who report randomly drops significantly, implying that they do not receive additional chances. Hence, even when an agent $a_i$ reports randomly at a time $t_i \to 0$ , it does not get a better reward than trustworthy agents.
    
\begin{lemma}\label{lemma:IC2}
    In \reform\ with RPTSC, an agent is incentivized to report truthfully when all the other agents choose trustworthy strategy. Assuming that the agents' beliefs satisfy:\footnote{Assumption $\mathsf{B}_2$ is implied from Assumption $\mathsf{B}$ in Prop.~\ref{prop:p2} for $n=2$.}
\begin{equation}\label{eqn:cond2}
    \mathsf{B}_2: \frac{p_{x_{i}}}{p_{y_{i}}} \geq \Delta_{i}
\end{equation}
\end{lemma}
\begin{proof} (SKETCH) For the proof, we first assume that the expected reward of an agent $a_i$ reporting truthfully is greater than any other strategy, i.e., $\mathbb{E}[R_i(x_{i}|x_{i});k=2] > \mathbb{E}[R_i(y_{i}|x_{i});k=2]$. Using Eqs.~\ref{eqn:cond1} and \ref{eqn:cond2}, we get $\left(p'_{x_{i}}-p_{x_{i}}\right) > -r$. We know that $\left(p'_{x_{i}}-p_{x_{i}}\right) \geq 0$ (from the self-predicting condition) and $-r\leq 0$. These impliy that $\mathsf{B}_2$ always holds. Thus, \reform\ with RPTSC reward scheme incentivizes truthful reporting.
\end{proof}

We conclude our game-theoretic analysis by proving that \reform\ with RPTSC  is NIC.

\begin{theorem}\label{thm:IC}
    \reform\ with RPTSC is strict Nash incentive compatible under assumptions $\mathsf{A}_1$ \emph{(Eq.~\ref{eqn:R1})} and $\mathsf{B}_2$ \emph{(Eq.~\ref{eqn:cond2})}.
\end{theorem}
\begin{proof}(SKETCH)
From Lemmas~\ref{lemma:IC1} and \ref{lemma:IC2}, we see that \reform\ with RPTSC reward incentivizes high efforts and truthful reporting. Moreover, we see that the expected reward $\mathbb{E}[R_i(y_{i}|x_{i});k=2]$ (Lemma~\ref{lemma:expreward}) is proportional to decay factor $\beta(\cdot)$ which decreases with an increase in time taken for reporting. Thus, we show that exerting efforts and early truthful reporting, i.e., trustworthy strategy is a strict Nash Equilibrium in \reform\ with RPTSC.
\end{proof}

\smallskip
\noindent\textit{Discussion.} With Lemma~\ref{claim:claim3}, we show that TERM is resistant to single report strategy. Further, RPTSC reward is also resistant to this strategy~\cite[Section 4.4]{PTSC}. One can observe that the reward in Eq. \ref{rptsc-reward} is zero in the case of single report strategy. Hence, for any appropriately chosen scalar constant $\alpha$, REFORM with RPTSC reward is also resistant to single report strategy.

In the next section, we show that REFORM with RPTSC reward guarantees significantly better fairness than RPTSC and satisfies qualitative fairness.
\subsection{Achieving Fairness through REFORM with RPTSC}

Consider the following propositions using our novel fairness definition $\gamma$-fairness (Definition~\ref{defn:gamma}).

    
    \begin{proposition}\label{ptsc-fair}
    For any $a_i\in\mathcal{A}$, RPTSC is $\gamma$-fair with $\frac{1}{\gamma} = \alpha\sum_{x_{i}\in\mathcal{X}} (1-q'_{x_{i}})\left(1 - (1 - q_{x_{i}})^{n-1}\right)$
    \end{proposition}
\begin{proof}(SKETCH) In RPTSC, the expected reward differs from optimal reward when reports do not match. That is, with probability ($1-q'_{x_{i}}$) it does not produce an optimal reward.
\end{proof}

\begin{proposition}\label{reform-fair}
For any $a_i\in\mathcal{A}$, \reform\ with RPTSC is $\gamma$-fair with $\frac{1}{\gamma} = \alpha\sum_{x_{i}\in\mathcal{X}} (1-rq'_{x_{i}})(1-q'_{x_{i}})\left(1 - (1 - q_{x_{i}})^{n-1}\right)$.
\end{proposition}
\begin{proof}(SKETCH) The reward in \reform\ with RPTSC, differs from optimal reward with smaller probability, i.e., $(1-q'_{x_{i}})(1-rq'_{x_{i}})$ because of the additional chance(s) given to a trustworthy agent.
\end{proof}

From the Props.~\ref{ptsc-fair} and \ref{reform-fair}, we observe that $\gamma$ in \reform\ with RPTSC is \emph{greater than} that in RPTSC. This implies that employing \reform\ guarantees an expected reward closer to the optimal reward. Thus, \reform\ with RPTSC is fairer for trustworthy agents compared to RPTSC.

To show that \reform\ with RPTSC is qualitatively fair, we use TERM scores of the agents for reputation scores in Definition~\ref{defn:qualitative}.
\begin{theorem}\label{thm:qf}
\reform\ with RPTSC satisfies qualitative fairness. 
\end{theorem}
\begin{proof}(SKETCH) Note that, \reform\ provides additional chances of pairing to the agents with high TERM scores. This reduces their chances of getting penalized from unfair pairing, leading to increased expected rewards. Thus, expected rewards are proportional to TERM scores, satisfying qualitative fairness.
\end{proof}

So far, we theoretically show that \reform\ improves fairness in PBMs through its solution of providing additional chance(s) to agents. We next validate our results through synthetic simulations.
\section{Experimental Evaluation}
In order to evaluate the performance of \reform\ with RPTSC reward, we simulate our crowdsourcing model. Our main objective is to empirically observe the fairness improvements of RPTSC with REFORM; thus, we neglect the decay factor used in the reward Eq.~\ref{eqn:reward}. We consider our setting with sufficient homogeneous tasks and answer space $\mathcal{X}= \{0,1,2\}$. To satisfy Assumptions $\mathsf{A}$ (Eq.~\ref{prop:p2}) and $\mathsf{A}_1$ (Eq.~\ref{eqn:R1}), we set $\alpha$ as 10 (for RPTSC) and 11 (for \reform).

We run our simulations for $200$ rounds, where each round has $n=50$ tasks and populated with $m=750$ agents. We assume that the set of agents comprise (i) 60\% of agents that follow trustworthy strategy (TAs) and report correct answers with a probability of at least 0.9; and (ii) 40\% of agents that report randomly (RAs) and that report every possible answer with equal probability. In each round, agents report their answers for a single task. We evaluate the agents against a randomly paired peer from the same task. In both the mechanisms, the sample size chosen for reward computation is equal to the number of tasks (i.e., 50). We compare rewards for TAs and RAs by averaging over all the rounds.

\begin{figure}[!t]
    \centering
    \includegraphics[width=\columnwidth]{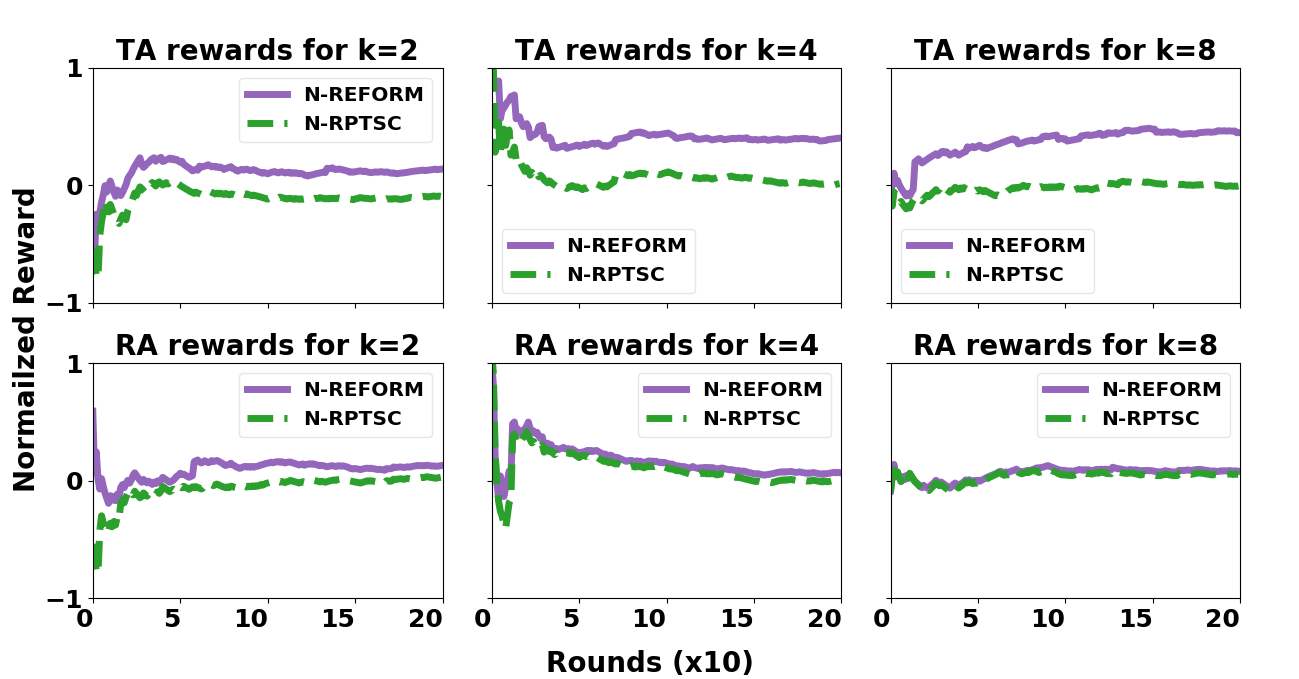}
    \caption{REFORM with RPTSC vs. RPTSC: Average rewards of agents using Trustworthy strategy (top row) and Random strategy (bottom row) as a function of round number}
    \label{fig:plot}
\end{figure}

Figure~\ref{fig:plot} shows the normalized rewards N-REFORM and N-RPTSC for TA and RA for $k=2, 4,\mbox{~and~} 8$. The rewards are defined as:
$$
\text{N-REFORM} = \frac{\text{REFORM-with-RPTSC-Reward}}{\text{Optimal-Reward} (M')},
$$
$$
\text{N-RPTSC}  = \frac{\text{RPTSC-Reward}}{\text{Optimal-Reward} (M')}.
$$

Here, N-REFORM and N-RPTSC are \textit{independent} of $\alpha$. Observe that with an increase in $k$, N-REFORM tends towards 1; that is, TA's reward in \reform\ tends to the optimal reward. Moreover, TA's reward in \reform\ is significantly higher than RPTSC, whereas the RA's rewards are \textit{almost the same}. Thus, one can note that \reform\ guarantees better rewards than RPTSC for TAs. We also observe $\gamma$ to be 0.09 and 0.05 for \reform\ and RPTSC, respectively after 200 rounds, which further highlight that \reform\ with RPTSC is fairer compared to RPSTC w.r.t. $\gamma$-fairness.

Since \reform\ produces higher rewards for TAs, one may observe that this increases the requester's budget. However, our experiments show that this increase is marginal. E.g., for $k=2$, \reform\'s budget per agent is approximately 5\% higher than RPTSC. Thus, \reform\ ensures better fairness for a moderate increase in the budget.

\section{Conclusion}

In this paper,  we focused on designing a fair reward scheme for crowdsourcing while incorporating temporal settings. Towards this, we introduced two notions of fairness, namely $\gamma$-fairness and quantitative fairness. To address the persistent issue of fairness in PBMs, we provided additional chances of pairing for trustworthy agents. To quantify an agent's trustworthiness, we introduced the reputation model, TERM,  and proved that it provides a high score for trustworthy reporting (Lemma~\ref{claim:claim1}). We proposed \reform\ a novel iterative framework that takes the reward scheme of any existing PBMs as a plug-in along with the reputation model, TERM. We proved that \reform\ with RPTSC is strict Nash incentive compatible (Theorem \ref{thm:IC}) and is resistant to single report strategy. We demonstrated that the framework \reform\ improves the fairness of RPTSC. We established that \reform\ with RPTSC achieves fairness with a marginal increase in the budget with the experiments.

\bibliographystyle{unsrt}  
\bibliography{aamas}

\newpage
\begin{appendix}
\section{Notations}
\begin{table}[h!]
\centering
\begin{tabular}{|p{3cm}|p{12cm}|}
\hline
Symbol &  Description\\
\hline\hline
$\mathcal{T}$ &  Set of tasks published in a round\\
\hline
$\mathcal{A}$ &  Set of agents in the system\\
\hline
$\mathcal{X}$ & Answer space\\
\hline
$x_{i}$ & Evaluation of agent $a_{i}$\\
\hline
$y_{i}$ & Report submitted by agent $a_{i}$\\
\hline
$t_{i}$ & Time taken by agent $a_{i}$ to report\\
\hline
$e_{i}$ & Effort an agent $a_{i}$ exerted for a task\\
\hline
$c(e_L), c(e_H)$ & Cost of low and high effort\\
\hline
$R_{i}(y_{i},t_{i})$ & Reward agent $a_{i}$ gets for reporting $y_{i}$ after time $t_{i}$\\
\hline
$P_{p}(x_{p})$ or $p_{x_p}$ & Probability that a random peer $a_{p}$ has evaluation equal to $x_{p}$ \\
\hline
$P_{p|i}(x_{p}|x_{i})$ or $p'_{x_p}$ & Probability that a random peer $a_{p}$ has evaluation $x_{p}$ when $a_{i}$'s evaluation is $x_{i}$ \\
\hline
$Q_{p}(x_{p})$ or $q_{x_p}$ & Probability that a random peer $a_{p}$ has report equal to $x_{p}$ \\
\hline
$Q_{p|i}(x_{p}|x_{i})$ or $q'_{x_p}$ & Probability that a random peer $a_{p}$ has report $x_{p}$ when  $a_{i}$'s evaluation is $x_{i}$ \\
\hline
$\Omega_{i,j}$ & TERM score of agent $a_{i}$ in the round $r_{j}$\\
\hline
$\mathcal{H}_{i,j}$ & History of the agent $a_i$ till the round $r_{j}$\\
\hline
$\phi_{i,j}$ & Round-score obtained by the agent $a_{i}$ for submitting the report in round $r_{j}$\\
\hline
$T_{p}(\Omega_{i},\Omega_{p})$ or $r$ & Probability that a random peer $a_{p}$ has TERM score less than $\Omega_i$\\
\hline
\end{tabular}
\caption{Notations}\label{tab:notations}
\end{table}

\section{Additional Proofs}
\subsection{Proof of Lemma~\ref{claim:claim3}}
\begin{proof}
            We observe that in TERM, the round-score ${\phi}_{i,j}=\frac{\mathbb{I}_{y_{i} = y_{p}}}{f(y_{i})t_{i}}$ directly depends on the report submitted $y_{i}$ and time taken $t_{i}$ by the agent $a_{i}$.
    
    Consider single report strategy where all the agents report the same answer $y_{i} = y$. In this case, $f(y_{i}) = 1$ and the expected round-score of the agent $a_{i}$ is,
    \begin{equation}
          CS: {\phi}_{i,j} = \frac{1}{t_{i}}\nonumber\\
    \end{equation}
    
    Suppose, out of $m$ agents in a round, $l$ agents report their evaluation $x$, and others report $y$.
    
    The expected round-score of agent $a_{i}$ who report $x$ in a non-colluding (trustworthy) strategy is,
    $$ TS:  {\phi}_{i,j} = \frac{l}{m}\times \frac{m}{l \times t_{i}} + \frac{m-l}{m} \times 0 = \frac{1}{t_{i}}$$
    
    We see that the expected round-score in colluding strategy, CS, is equal to a trustworthy strategy, TS. Therefore, any rational agent prefers to choose a trustworthy strategy, as it does not benefit from single report strategy. Thus, we claim that TERM is resistant to single report strategy.
\end{proof}

\subsection{Proof of Lemma~\ref{claim:claim1}}

\begin{proof}
        To prove the lemma, we show that TERM produces high scores for reporting truth early after exerting efforts, considering that all other agents are trustworthy. Agents exerting effort is assured by Lemma \ref{lemma:IC1}; this overcomes random reporting. We have seen that, TERM score obtained by agent $a_{i}$ is $\Omega_{i,j} = G({\psi}_{i,j})$ in round $r_{j}$. Trivially, TERM score increases with an increase in cumulative-score, which aggregates all the normalized round-scores. Hence, an increase in round-score increases the TERM score. 
        
        Assuming two agents with identical round-scores in previous rounds, a difference in round-score of the present round will show a difference in their TERM score. From Eq. \ref{eqn:r-s}, we calculate the round-score as ${\phi}_{i,j}=\frac{\mathbb{I}_{y_{i} = y_{p}}}{f(y_{i})t_{i}}$. Here, $f(y_{i})$ is the frequency function of $y_{i}$, calculated as the ratio of the number of reports (say, $b + 1$) that match with report $y_{i}$ to the total number of sampled reports (say, $n$). That is, $f(y_{i}) = \frac{num(y_{i})}{\sum_{y\in \mathcal{X}}num(y)} = \frac{b+1}{n}$. Further, $y_{p}$ is the random report sampled from the same task. We have seen that a trustworthy agent's strategy is ($x_{i}$, $e_{H}$, $t_i^*$), where $t_i^*$ is time taken for solving the task. Round-score of a agent who solves the task, but does not report truth (i.e., with strategy ($y_{i}$, $e_{H}$, $t_{i}$), where $t_{i} \geq t^*_{i}$) is,
    
    \begin{equation}
    \begin{split}
            &\frac{p'_{y_{i}}}{t_{i}} \sum_{b=0}^{n-1}{n-1 \choose b}{(p_{y_{i}})}^{b} (1-p_{y_{i}})^{n-b-1} \frac{n}{b+1}  \\
            & = \frac{p'_{y_{i}}}{t_{i}}  \sum_{b=0}^{n-1}{n \choose b+1}{(p_{y_{i}})}^{b} (1-p_{y_{i}})^{n-b-1}  \\
            & = \frac{p'_{y_{i}}}{p_{y_{i}}\times t_{i}} \sum_{b=1}^{n}{n \choose b}{(p_{y_{i}})}^{b} (1-p_{y_{i}})^{n-b} \nonumber \\
            & = \frac{p'_{y_{i}}}{p_{y_{i}}\times t_{i}} (1 - (1-p_{y_{i}})^{n}) \\
            & \leq \frac{p'_{x_{i}}}{p_{x_{i}}\times t_{i}} (1 - (1-p_{x_{i}})^{n}) \quad(\text{From, Eq. \ref{eqn:A1}, \ref{eqn:A2}})\\
            & \leq \frac{p'_{x_{i}}}{p_{x_{i}}\times t^*_{i}} (1 - (1-p_{x_{i}})^{n}) \quad (\text{Since, }t_{i} \geq t^*_{i})
    \end{split}
    \end{equation}
    
    From the first inequality, we observe that the round-score of an agent when it reports truth (i.e., its evaluation $x_i$) is greater. And it is evident that the round-score increases with early reporting. Hence, TERM incentivizes early as well as truthful reporting.
\end{proof}

\subsection{Proof of Lemma~\ref{lemma:expreward}}

\begin{proof}
        Observe that when $q_{y_{i}} = 0$ (i.e., the probability with which agent $a_{i}$'s peer reports $y_{i}$ is $0$), the expected reward of agent $a_{i}$ is $0$. Now, consider agent $a_{i}$ with  $q_{y_{i}} > 0$ and evaluation $x_{i}$, reports $y_{i}$ after time $t_{i}$. From Proposition \ref{prop:p1}, the expected reward of an agent $a_{i}$ for reporting $y_{i}$  in RPTSC is $E'$.
        
        
        \noindent The expected reward, $\mathbb{E}[R_i(y_{i}|x_{i});k=2]$, is calculated as follows. For this, let TERM score of the agent $a_{i}$ in round $r_{j}$ be $\Omega_{i}$. Let agent $a_{p}$ with report $y_{p}$ and TERM score $\Omega_{p}$ be the peer against whom $a_{i}$ is evaluated. As seen before, agent $a_{i}$'s belief regarding the TERM scores of the any peer is same, i.e., $\forall a_{p} \in \mathcal{A}, r = T_{p}(\Omega_i,\Omega_{p})$.
        
        From Algorithm~\ref{algo::reform-rptsc}, if the TERM score $\Omega_{i}$ of agent $a_{i}$ is less than $\Omega_{p}$, in the first chance of pairing, then agent $a_{i}$ does not get an additional chance to pair. In this case, the expected reward is the same as RPTSC expected reward times the decay factor, i.e., $\beta(t_{i}) \times E'$. 
        
        However, if $a_{i}$'s TERM score is higher than that of its peer's, then since $k=2$, it gets another chance to pair. In this case, we have: (i) if $y_{i} = y_{p}$ the reward is equal to optimal reward times with the decay factor, i.e., $\beta(t_{i})\times M'$; and (ii) if $y_{i} \not= y_{p}$ the expected reward is equal to $\beta(t_{i})\times E'$, as it receives an additional chance. Formally, we have,
        
        \begin{equation}\nonumber
            \begin{split}
                & \mathbb{E}[R_i(y_{i}|x_{i});k=2]  \\
                & = Pr(\Omega_{i} < \Omega_{p})\beta(t_{i}) E' + Pr(\Omega_{i} > \Omega_{p})\\
                &\quad \bigg(Pr(y_{i} = y_{p}|x_{i}) \beta(t_{i}) M' + Pr(y_{i} \neq y_{p}|x_{i}) \beta(t_{i}) E'\bigg)\\
                & = \beta(t_{i}) (1-r) E' + \beta(t_{i})r\big(q'_{y_{i}} M' + (1-q'_{y_{i}}) E'\big)\\
                & = \beta(t_{i}) \bigg(E' + r q'_{y_{i}}\left(M' - E'\right)\bigg)\\
                & = \alpha\beta(t_{i})\bigg(\left(\frac{q'_{y_{i}}}{q_{y_{i}}}-1\right) + r q'_{y_{i}} \frac{(1-q'_{y_{i}})}{q_{y_{i}}}\bigg)\left(1 - \left(1 - q_{y_{i}}\right)^{n-1}\right) \\
            \end{split}
        \end{equation}
        This completes the proof of the lemma.        
\end{proof}

\subsection{Proof of Corollary~\ref{cor::k}}

\begin{proof}
        Similar to the proof given for Lemma \ref{lemma:expreward}, the expected reward of an agent $a_{i}$ with evaluation $x_{i}$ and report $y_{i}$ in \reform\ with RPTSC is, 
    \begin{equation}
        \begin{split}
            & \mathbb{E}[R_i(y_{i}|x_{i});k] = (1-r)E' + r\bigg(q'_{y_{i}}M' + (1-q'_{y_{i}})\\
            & \quad \bigg((1-r)E' +\ldots r\big(q'_{y_{i}}M' + (1-q'_{y_{i}})E'\big)\bigg)\bigg) \nonumber\\
            & = rq'_{y_{i}}M'\bigg( 1 + (r-rq'_{y_{i}})+\ldots+(r-rq'_{y_{i}})^{k-1}\bigg) + E'(r-rq'_{y_{i}})^{k-1} \\
            &\quad + E'(1-r)\bigg(1 + (r-rq'_{y_{i}}) + \ldots + (r-rq'_{y_{i}})^{k-2}\bigg) \\
            & = rq'_{y_{i}}M'\sum_{i = 1}^{k} \left(r-rq'_{y_{i}}\right)^{i-1} + E'\bigg((r-rq'_{y_{i}})^{k-1} + (1-r)\sum_{i = 2}^{k} (r-rq'_{y_{i}})^{i-2}\bigg)\\
        \end{split}
    \end{equation}
    
    From the above, we see that every term is positive, and with an increase in $k$ expected reward increases. This proves the lemma.
\end{proof}

\subsection{Proof of Lemma~\ref{lemma:random}}

\begin{proof}
        Note that a random agent does not exert efforts for a task, i.e., its evaluation for the task is $\varnothing$. Therefore, $p'_{y_{i}}= p_{y_{i}}$. Now the expected reward of a random agent, when all agents are trustworthy (From Lemma \ref{lemma:expreward}), is as follows,
        \begin{equation}
        \begin{split}
            E_{ra} & = \alpha\beta(t_{i})\bigg(\left(\frac{p_{y_{i}}}{p_{y_{i}}}-1\right) + r(1-p_{y_{i}})\frac{p_{y_{i}}}{p_{y_{i}}}\bigg)\left(1-(1-p_{y_{i}})^{n-1}\right) \nonumber\\
            & = \alpha\beta(t_{i})r(1-p_{y_{i}})\left(1-(1-p_{y_{i}})^{n-1}\right)
        \end{split}
        \end{equation} 
        The equality proves the lemma.
\end{proof}

\subsection{Proof of Lemma~\ref{lemma:IC1}}

\begin{proof}
     Before evaluation of the task, agent $a_{i}$'s expected utility for investing high efforts is $\overline{Ref}_{i}(\alpha)-c(e_{H})$ and its expected utility when it reports randomly is $E_{ra}-c(e_{L})$.\\
        We show that $\overline{Ref}_{i}(\alpha) - c(e_{H}) > E_{ra}-c(e_{L})$, to prove that \reform\ incentivizes high efforts.

        \begin{equation}\label{q1}
            \begin{split}
                \overline{Ref}_{i}(\alpha) & \geq c(e_{H})-c(e_{L}) + \alpha \beta(t_{i})\quad (\text{From Eq. \ref{eqn:R1}})\\
              \implies& \overline{Ref}_{i}(\alpha) - c(e_{H})  \geq \alpha \beta(t_{i}) - c(e_{L})
            \end{split}
        \end{equation}
        
        From Lemma \ref{lemma:random}, we have the expected reward of random agent:
        \begin{equation}\label{q2}
            \begin{split}
                E_{ra} &= r\alpha\beta(t_{i})\mathbb{E}_{y_{i}\in \mathcal{X}}\big[(1-p_{y_{i}}) \left(1 - \left(1 - p_{y_{i}}\right)^{n-1}\right)\big] \\
                &= \alpha\beta(t_{i}) \sum_{y_{i}\in \mathcal{X}}\big[rp_{y_{i}}(1-p_{y_{i}}) \left(1 - \left(1 - p_{y_{i}}\right)^{n-1}\right)\big]\\
                & < \alpha\beta(t_{i})
            \end{split}
        \end{equation}
        
        The last inequality is because $\sum_{y_{i}} p_{y_{i}} =1$, and $0 < r(1-p_{y_{i}}) \left(1 - \left(1 - p_{y_{i}}\right)^{n-1}\right) < 1 \quad \forall y_{i}\in \mathcal{X}$. Then,
        \begin{equation}\nonumber
            \begin{split}
                p_{y_{i}}r (1-p_{y_{i}}) \left(1 - \left(1 - p_{y_{i}}\right)^{n-1}\right) &< p_{y_{i}} \quad \forall y_{i}\\
                \sum_{y_{i}} p_{y_{i}}r (1-p_{y_{i}}) \left(1 - \left(1 - p_{y_{i}}\right)^{n-1}\right) &< \sum_{y_{i}} p_{y_{i}}\\
                \sum_{y_{i}} p_{y_{i}}r (1-p_{y_{i}}) \left(1 - \left(1 - p_{y_{i}}\right)^{n-1}\right) &< 1
            \end{split}
        \end{equation}
        
        From the above two inequalities \ref{q1} and \ref{q2}, we have,
        \begin{equation}\nonumber
            \begin{split}
                \overline{Ref}_{i}(\alpha) - c(e_{H})  & \geq \alpha \beta(t_{i}) - c(e_{L})\\
                & > E_{ra} - c(e_{L})\\
            \end{split}
        \end{equation}
        
        The expected utility of an agent before evaluation for exerting efforts is strictly greater than the expected utility from random reporting. Thus, a random agent is incentivized to exert high efforts in \reform.   
\end{proof}

\subsection{Proof of Lemma~\ref{lemma:IC2}}

\begin{proof}
         Consider agent $a_{i}$ with evaluation $x_{i}$ and report $y_{i}$ submitted after time $t_{i}$. We assume that all other agents are trustworthy. For the strategy profile where all the agents are trustworthy, $q' = p'; q=p$. 
    
    From Lemma \ref{lemma:expreward}, the expected reward of agent $a_{i}$ for reporting $y_{i}$, in \reform\ with RPTSC when $k=2$ is $\mathbb{E}[R_i(y_{i}|x_{i});k=2]$.
    
    We prove that the expected reward of a strategic agent, $a_{i}$ is less when it reports any value other than its evaluation $x_{i}$. For this, we start with the assumption that the reward for reporting the truth is more than the reward for reporting non-truth and arrive at a noticeably obvious result.
    
    \begin{equation*}
        \begin{split}
            \implies & \mathbb{E}[R_i(x_{i}|x_{i});k=2] > \mathbb{E}[R_i(y_{i}|x_{i});k=2]\\
            \implies & \alpha\beta(t_{i})\left[\left(\frac{p'_{x_{i}}}{p_{x_{i}}}-1\right) + r(1-p'_{y_{i}})\frac{p'_{x_{i}}}{p_{x_{i}}}\right]\left(1-(1-p_{x_{i}})^{n-1}\right) \\
            & \quad > \alpha\beta(t_{i})\left[\left(\frac{p'_{y_{i}}}{p_{y_{i}}}-1\right) + r(1-p'_{y_{i}})\frac{p'_{y_{i}}}{p_{y_{i}}}\right]\left(1-(1-p_{y_{i}})^{n-1}\right) \\
            \implies &\left(\frac{p'_{x_{i}}}{p_{x_{i}}}-1\right)\left(1-(1-p_{x_{i}})^{n-1}\right) - \left(\frac{p'_{y_{i}}}{p_{y_{i}}}-1\right) \left(1-(1-p_{y_{i}})^{n-1}\right) \\
            & \quad > r\bigg[ \frac{(1-p'_{y_{i}}) p'_{y_{i}}}{p_{y_{i}}} \left(1-(1-p_{y_{i}})^{n-1}\right) -\\ 
            & \quad\frac{p'_{x_{i}}(1-p'_{x_{i}})}{p_{x_{i}}} \left(1-(1-p_{x_{i}})^{n-1}\right)\bigg]\\
            \implies &\left(\frac{p'_{x_{i}}}{p_{x_{i}}}-1\right)\bigg[\left(1-(1-p_{x_{i}})^{n-1}\right) - \frac{p_{x_{i}}}{p_{y_{i}}} \left(1-(1-p_{y_{i}})^{n-1}\right)\bigg] \\
            & \quad > r\bigg[ \frac{(1-p'_{y_{i}})p'_{y_{i}}}{p_{y_{i}}} \left(1-(1-p_{y_{i}})^{n-1}\right) - \\ 
            &\quad \frac{p'_{x_{i}}(1-p'_{x_{i}})}{p_{x_{i}}} \left(1-(1-p_{x_{i}})^{n-1}\right)\bigg]
            \quad \left(\text{From, Eq. \ref{eqn:cond1} \& \ref{eqn:cond2}}\right)\\
            \implies &\left(p'_{x_{i}}-p_{x_{i}}\right)\bigg[p_{y_{i}}\left(1-(1-p_{x_{i}})^{n-1}\right) - p_{x_{i}} \left(1-(1-p_{y_{i}})^{n-1}\right)\bigg] \\
            & \quad > r\bigg[(1-p'_{y_{i}})p'_{y_{i}}p_{x_{i}} \left(1-(1-p_{y_{i}})^{n-1}\right) -\\
            & \quad p'_{x_{i}}{p_{y_{i}}}(1-p'_{x_{i}}) \left(1-(1-p_{x_{i}})^{n-1}\right)\bigg]\\
            \implies &\left(p'_{x_{i}}-p_{x_{i}}\right)\bigg[p_{y_{i}}\left(1-(1-p_{x_{i}})^{n-1}\right) - p_{x_{i}} \left(1-(1-p_{x_{i}})^{n-1}\right)\bigg] \\
            & > r\bigg[p_{x_{i}} \left(1-(1-p_{y_{i}})^{n-1}\right) -
            p'_{x_{i}}{p_{y_{i}}}(1-p'_{x_{i}})\left(1-(1-p_{x_{i}})^{n-1}\right)\bigg]
        \end{split}
    \end{equation*}

    \begin{equation*}
        \begin{split}
             &\quad \left(\text{Since, }(1-p'_{y_{i}})p'_{y_{i}} \leq 1\right)\\
            \implies & -\frac{p_{x_{i}} \left(1-(1-p_{y_{i}})^{n-1}\right) -
            p'_{x_{i}}{p_{y_{i}}}(1-p'_{x_{i}}) \left(1-(1-p_{x_{i}})^{n-1}\right)}
            {p_{x_{i}} \left(1-(1-p_{y_{i}})^{n-1}\right) - p_{y_{i}}\left(1-(1-p_{x_{i}})^{n-1}\right)}\\
            & \quad < \frac{\left(p'_{x_{i}}-p_{x_{i}}\right)}{r} \left(\text{Since, numerator $>$ denominator}\right)\\
            \implies &\frac{\left(p'_{x_{i}}-p_{x_{i}}\right)}{r} > -1 \\
            \implies &\left(p'_{x_{i}}-p_{x_{i}}\right) > -r 
        \end{split}
    \end{equation*}
        
    We know that $p'_{x_{i}}-p_{x_{i}} > 0$ (From, self-predicting condition) and $-r \leq 0$. Thus, the above condition is true, implying the assumption made is true. That is, the expected utility for reporting the truth is strictly more than strategic reporting. Hence, proving the lemma.
\end{proof}

\subsection{Proof of Theorem~\ref{thm:IC}}
\begin{proof}
     From Lemmas \ref{lemma:IC1} and \ref{lemma:IC2}, we see that the agent $a_{i}$ under the given assumptions is strictly incentivized to exert efforts and report truthfully. We know that the expected reward of an agent $a_{i}$ for reporting the truth, i.e., $x_{i}$ is,
            $$\alpha\beta(t_{i})\bigg(\frac{p'_{x_{i}}}{p_{x_{i}}}-1 + rp'_{x_{i}}\frac{(1-p'_{x_{i}})}{p_{x_{i}}}\bigg)\left(1 - \left(1 - p_{x_{i}}\right)^{n-1}\right)$$
    One can observe that this expected reward is proportional to $\beta(t_{i})$ and $\beta(\cdot)$ function increases with decrease in time $t_{i}$. Thus, in \reform\ with RPTSC trustworthy strategy is a strict Nash equilibrium. And hence, \reform\ with RPTSC is strict Nash incentive compatible.
\end{proof}

\subsection{Proof of Proposition~\ref{prop:p1}}

\begin{proof}
         From Proposition \ref{prop:p1} and Eq. \ref{eqn:optreward}, expectation of the difference in optimal reward ($M'$) and expected reward ($E'$) for RPTSC over all possible evaluations is,
        
        \begin{equation}
            \begin{split}
                & \mathbb{E}_{x_i \in \mathcal{X}}[M'-E'] \\
                & = \mathbb{E}_{x_i \in \mathcal{X}}\bigg[\alpha \left(\frac{1}{q_{x_{i}}} - 1\right)\left(1 - \left(1 - q_{x_{i}}\right)^{n-1}\right) \\
                &\quad - \alpha\left(\frac{q'_{x_{i}}}{q_{x_{i}}} - 1\right)\left(1 - \left(1 - q_{x_{i}}\right)^{n-1}\right)\bigg] \\
                & = \alpha\mathbb{E}_{x_i \in \mathcal{X}}\left[\left(\frac{1-q'_{x_{i}}}{q_{x_{i}}}\right)\left(1 - \left(1 - q_{x_{i}}\right)^{n-1}\right)\right]\nonumber\\
                & = \alpha\sum_{x_{i}\in\mathcal{X}}q_{x_{i}}\left(\frac{1-q'_{x_{i}}}{q_{x_{i}}}\right)\left(1 - \left(1 - q_{x_{i}}\right)^{n-1}\right)\\
                & = \alpha\sum_{x_{i}\in\mathcal{X}}\left(1-q'_{x_{i}}\right)\left(1 - \left(1 - q_{x_{i}}\right)^{n-1}\right)
            \end{split}
        \end{equation}
        
        From the above, we see that RPTSC is\\ $\alpha\sum_{x_{i}\in\mathcal{X}}\left(1-q'_{x_{i}}\right)\left(1 - \left(1 - q_{x_{i}}\right)^{n-1}\right)$-fair.
\end{proof}

\subsection{Proof of Proposition~\ref{prop:p2}}

\begin{proof}
     From Lemma \ref{lemma:expreward}, we have the expected reward of \reform\ with $\forall t, \beta(t)=1$ as, $\mathbb{E}[R_i(x_{i}|x_{i});k=2]$
        
    Optimal reward in \reform\ with RPTSC is $M'$.
    \begin{equation}
        \begin{split}
            & = \mathbb{E}_{x_i \in \mathcal{X}}[M' - \mathbb{E}[R_i(x_{i}|x_{i});k=2]] \\
            & = \mathbb{E}_{x_i \in\mathcal{X}}\bigg[\alpha\left(\frac{1}{q_{x_{i}}} - 1\right)\left(1 - \left(1 - q_{x_{i}}\right)^{n-1}\right)\\ 
            &\quad - \alpha\left(\left(\frac{q'_{x_{i}}}{q_{x_{i}}}-1\right) + r(1-q'_{x_{i}})\frac{q'_{x_{i}}}{q_{x_{i}}}\right)\left(1-(1-q_{x_{i}})^{n-1}\right)\bigg] \\
            & = \alpha\mathbb{E}_{x_i \in\mathcal{X}}\left[\left(\frac{1-q'_{x_{i}}}{q_{x_{i}}} - r(1-q'_{x_{i}})\frac{q'_{x_{i}}}{q_{x_{i}}}\right)\left(1 - \left(1 - q_{x_{i}}\right)^{n-1}\right)\right]\nonumber\\
            & = \alpha\mathbb{E}_{x_i \in\mathcal{X}}\left[(1-rq'_{x_{i}})\frac{1-q'_{x_{i}}}{q_{x_{i}}}\left(1 - (1 - q_{x_{i}})^{n-1}\right)\right]\\
            & = \alpha\sum_{x_{i}\in\mathcal{X}}\left[q_{x_{i}}(1-rq'_{x_{i}})\frac{1-q'_{x_{i}}}{q_{x_{i}}}\left(1 - (1 - q_{x_{i}})^{n-1}\right)\right]\\
            & = \alpha\sum_{x_{i}\in\mathcal{X}} (1-rq'_{x_{i}})(1-q'_{x_{i}})\left(1 - (1 - q_{x_{i}})^{n-1}\right)\\
        \end{split}
    \end{equation}
    From the above, we see that \reform\ is $\alpha\sum_{x_{i}\in\mathcal{X}} (1-rq'_{x_{i}})(1-q'_{x_{i}})\left(1 - (1 - q_{x_{i}})^{n-1}\right)$-fair.
\end{proof}

\subsection{Proof of Theorem~\ref{thm:qf}}

\begin{proof}
    The expected reward Lemma \ref{lemma:expreward} of a agent $a_{i}$ for reporting $y_{i}$ is 
    \begin{equation}
        \begin{split}
            & \mathbb{E}[R_i(x_{i}|x_{i});k=2]\\
            & = \alpha\beta(t_{i})\left[\left(\frac{q'_{y_{i}}}{q_{y_{i}}}-1\right) + r(1-q'_{y_{i}})\frac{q'_{y_{i}}}{q_{y_{i}}}\right]\left(1-(1-q_{y_{i}})^{n-1}\right) \nonumber\\
            & = \beta(t_{i})E' + \alpha\beta(t_{i}) r(1-q'_{y_{i}})\frac{q'_{y_{i}}}{q_{y_{i}}}\left(1-(1-q_{y_{i}})^{n-1}\right)
        \end{split}
    \end{equation}
    Where $r$ is the probability with which the peer's TERM score is less than that of agent $a_{i}$'s score.
    
    Consider two agents $a_{1},a_{2}$ who reported same answers, with TERM scores $\Omega_{1},\Omega_{2}$ (where, $\Omega_{1}<\Omega_{2}$) respectively. The beliefs about other agents having lesser TERM scores is given by $T_{p}$, where $T_{p}(\Omega,\Omega_{p})$ is the probability with which peer $a_{p}$'s TERM score is less than $\Omega$. Since, $$\Omega_{1}<\Omega_{2} \implies T_{p}(\Omega_{1},\Omega_{p}) < T_{p}(\Omega_{2},\Omega_{p})$$
    
    Assuming that both the agents have the same beliefs. The expected reward for the agents $a_{1}$ and $a_{2}$ are given as,
    \begin{equation}
        \begin{split}
            E_{1} = \beta(t_{i})E' + & \alpha\beta(t_{i})T_{p}(\Omega_{1},\Omega_{p}) (1-q'_{y_{i}})\frac{q'_{y_{i}}}{q_{y_{i}}}\left(1-(1-q_{y_{i}})^{n-1}\right)  \\
            E_{2} = \beta(t_{i})E' + & \alpha \beta(t_{i}) T_{p}(\Omega_{2},\Omega_{p}) (1-q'_{y_{i}})\frac{q'_{y_{i}}}{q_{y_{i}}}\left(1-(1-q_{y_{i}})^{n-1}\right)  \nonumber\\
            E_{1} &< E_{2}
        \end{split}
    \end{equation}
    
    We see that the agent's expected reward with a greater reputation is higher than the agent's expected reward with a lower reputation having the same report. Therefore, \reform\ with RPTSC is qualitatively fair.
\end{proof}


\end{appendix}

\end{document}